\documentclass[a4paper]{article}
\usepackage[a4paper,margin=0.8in]{geometry}
\usepackage[boxed,linesnumbered,vlined]{algorithm2e}
\usepackage{graphicx}
\usepackage{listings}
\usepackage{amsthm}
\usepackage{amsmath}
\usepackage{amssymb}
\usepackage{relsize}
\usepackage{MnSymbol,wasysym}
\usepackage[export]{adjustbox}
\usepackage{mdframed}
\usepackage{mathtools}
\usepackage[dvipsnames]{xcolor}
\usepackage{soul,url}
\usepackage{cleveref}
\usepackage{tikz}
\usetikzlibrary{arrows.meta, positioning}
\usepackage{wrapfig}

\newtheorem{definition}{Definition}
\newtheorem{theorem}{Theorem}
\newtheorem{lemma}[theorem]{Lemma}
\newtheorem{corollary}[theorem]{Corollary}

\newcommand{\remove}[1]{}

\newcommand{\etal}{\emph{et al.}\xspace}
\newcommand{\eg}{\emph{e.g.,}\xspace}

\newcommand{\ie}{\emph{i.e.,}\xspace}

\SetKwComment{Comment}{/* }{ */}
\SetKwProg{Fn}{function}{}{}
\SetKwInput{GlobalData}{Global Automaton}
\SetKwBlock{Upon}{Upon}{end}

\newcounter{claimcounter}

\newcounter{lemmacounter}

\newcounter{theoremcounter}

\renewenvironment{proof}
    {\par\noindent\textit{Proof.} }
    {\hfill$\square$\par}

\usepackage{comment}

\date{April 2025}

\title{Self-Stabilizing Replicated State Machine  Coping with Byzantine and Recurring Transient Faults}

\date{}

\author{
   Shlomi Dolev\footnote{Ben-Gurion University of the Negev, Israel. Email: dolev@bgu.ac.il} \quad 
   Amit Hendin\footnote{Ben-Gurion University of the Negev, Israel. Email: hendina@post.bgu.ac.il} 
   \quad
   Maurice Herlihy\footnote{Brown University, USA. Email:  mph@cs.brown.edu} \quad
   Maria Potop Butucaru\footnote{Sorbonne Université, LIP6, Paris, France. Email: maria.potop-butucaru@lip6.fr} \quad
   Elad Michael Schiller\footnote{Chalmers University of Technology, Sweden. Email: elad.schiller@chalmers.se }
}

\begin{document}
\maketitle

%\usepackage{authblk}

%\author[1]{Shlomi Dolev}
%\author[2]{Amit Hendin}
%\author[3]{Maurice Herlihy}
%\author[4]{Maria Potop Butucaru}
%\author[5]{Elad Michael Schiller}

%\affil[1]{Ben-Gurion University of the Negev, Israel}
%\affil[2]{Brown University, USA}
%\affil[3]{Sorbonne Université, LIP6, Paris, France}
%\affil[4]{Chalmers University of Technology, Sweden}

%\ccsdesc[500]{Theory of computation~Distributed computing models}
%\keywords{Consensus, Self-Stabilization, Byzantine Tolerance} 
%\funding{}

\begin{abstract}
The ability to perform repeated Byzantine agreement lies at the heart of
important applications such as blockchain price oracles or replicated
state machines. Any such protocol requires the following properties:
(1) \textit{Byzantine fault-tolerance}, because not all participants can be
assumed to be honest, (2) r\textit{ecurrent transient fault-tolerance}, because
even honest participants may be subject to transient ``glitches'', (3)
\textit{accuracy}, because the results of quantitative queries (such as price
quotes) must lie within  the interval of honest participants' inputs, and (4)
\textit{self-stabilization}, because it is infeasible to reboot a distributed system following a fault.

This paper presents the first protocol for repeated Byzantine agreement
that satisfies the properties listed above.
Specifically, starting in an arbitrary system configuration, our protocol establishes consistency. It preserves consistency in the face of up to $\lceil n/3 \rceil -1$ Byzantine participants {\em and} constant recurring (``noise'') transient faults, of up to $\lceil n/6 \rceil-1$ additional malicious transient faults, or even more than $\lceil n/6 \rceil-1$ (uniformly distributed) random transient faults, in each repeated Byzantine agreement. 
\end{abstract}

%\keywords{Self-Stabilization, Replicated State Machine, Byzantine Faults, Transient Faults, Fault Tolerance, Distributed Computing, Security}

\section{Introduction}
In finance, a \emph{price oracle} (or just \textit{oracle}) is a service that connects users (often smart contracts on blockchains) with timely real-world data.
Examples might include the spot price of Bitcoin, the current dollar/euro exchange rate, or nightly temperatures in Florida citrus groves. Today, oracle services are a growing business sector,
led by companies such as Chainlink~\cite{chainlink},
Pyth Network~\cite{pyth}, and others~\cite{oraclesurvey}.

Much hinges on price oracle accuracy.
Oracle price data drives automated trading algorithms,
and unexpected fluctuations may trigger ruinous margin calls or liquidations.
To avoid costly errors or fraud, commercial oracle services typically employ a variety of ad-hoc techniques, typically aggregating and smoothing data from multiple sources~\cite{oraclesurvey}.

This paper introduces a class of Byzantine agreement algorithms
well-suited to provide a systematic,
rigorous foundation for \emph{decentralized} price oracles
immune to tampering from a small coalition of participants.
A decentralized price oracle consists of a set of $n$ processes, each connected to a distinct data source, such as a stock market feed or a sensor.
A user query causes the processes to reach consensus on an aggregate price to be delivered to that user.
Here is a list of (informally stated) properties a consensus protocol for a price oracle
should satisfy.

\begin{itemize}
\item
  \textbf{Byzantine Fault-Tolerance}:
  Because the stakes are high, one cannot assume all participants are honest.
  Our protocol tolerates $\lceil n/3 \rceil - 1$ Byzantine
  (permanently dishonest) processes.
\item
  \textbf{Recurrent Transient Fault-Tolerance}:
  Even honest participants may be subject to recurring \emph{transient} errors, where an otherwise honest participant periodically undergoes an incorrect state change, but then resumes honest behavior.
  Our protocol tolerates $\lceil n/6 \rceil -1$ recurring intermittent transient errors.
\item
  \textbf{Price Accuracy}:
  As conditions change,
  honest participants may receive different prices from their sources,
  while dishonest participants may claim arbitrary prices.
  Even in the face of differing inputs and disruptive participants,
  it is not acceptable to deliver a null $\bot$ price to a user.
  Instead,  our protocol guarantees that the price delivered to the user
  lies within the \emph{range} of prices proposed by honest participants.
  (Under some circumstances,
  the protocol can be adapted to return specific values,
  such as the internal mode or median.)
\item
  \textbf{Self-Stabilization}:
  When something goes wrong, it is not feasible to restart a decentralized service.
  An oracle state is \emph{legitimate} if it could have
  been produced by a failure-free execution.
  Our protocol is \emph{self-stabilizing}: starting from any (possibly illegitimate) state, the system will eventually enter a legitimate state,
  and remain in legitimate states
  until the next transient failure, see \cite{DBLP:books/mit/Dolev2000}.
\end{itemize}
These properties are useful for applications beyond price oracles.
For example, a \emph{replicated state machine} is a fundamental task in distributed computing in which multiple processes coordinate to act as a single process (state machine).
Devising replicated state machine protocols that work in the face of imperfect communication and Byzantine participants is a classic problem of distributed
computing~\cite{DBLP:journals/tocs/JhaBGMSTRZB19,nakamoto2008bitcoin}. 

 \noindent
{\bf Related work.} 
 Prior work on Byzantine agreement has addressed tolerating Byzantine participants and faults,
 \eg \cite{DBLP:journals/jacm/DolevW04, DBLP:journals/tdsc/BrownsteinDK22}. 
 By contrast, our work here investigates the effects of \textit{repeated} transient faults in the scope of \textit{repeated} consensus to realize a replicated state machine. 
Our protocol has more flexibility in choosing decision values than protocols restricted to approximate agreement \cite{DBLP:journals/jacm/DolevLPSW86},  median \cite{DBLP:conf/opodis/StolzW15}, or interval agreement 
\cite{DBLP:conf/srds/MelnykW18}. 

Stolz and Wattenhofer~\cite{DBLP:conf/opodis/StolzW15} propose a Byzantine agreement protocol that satisfies median validity, meaning that the non-faulty processes decide on a value that is at most $t$ places away from the median value of the non-faulty nodes, where $t$ is the upper limit of the number of Byzantine processes. Melnyk and Wattenhofer~\cite{DBLP:conf/srds/MelnykW18} propose a Byzantine agreement protocol that satisfies interval validity, which requires the non-faulty processes to decide values close to the $k^{th}$ smallest initial value of the non-faulty processes. While both protocols take the median of a (possibly not identical) vector of the inputs structured by each process, here we ensure that each non-Byzantine process structures an {\em identical} vector of the inputs. While these two works aim to approximate almost equal (but possibly different) values decided by non-Byzantine participants, our approach obtains consistency on an entirely agreed-upon vector.
All non-Byzantine participants agree upon each vector entry before taking the median. In order to obtain such an identical vector, another Byzantine agreement algorithm, that may return $\bot$, is employed; we chose a classical one presented in \cite{TURPIN198473}.  

Binum et. al. \cite{DBLP:conf/sss/BinunCDKLPYY16} presented a self-stabilizing Byzantine replicated machine; however, they do not address the tolerance of recurring transient faults nor ensure strong validity, namely, the usage of the preceding state replica under one-third Byzantine and one-sixth recurring transient faults. Several approaches were suggested to seamlessly cope with transient faults, including superstabilization for tolerating the dynamicity of the communication graph \cite{DBLP:journals/cjtcs/DolevH97}, fault containment, and local stabilization for coping with a close-by corruption of the state of the participants, 
\cite{DBLP:journals/dc/GhoshGHP07, DBLP:journals/jpdc/AfekD02}. 
%In addition, tolerating recurring transient faults by preservation of Hamming distance throughout the operation in the scope of computer architecture is suggested in  \cite{DBLP:journals/et/DolevFTS13}. 
Here, we introduce the possibility of tolerating Byzantine and recurrent (arbitrary and randomized) transient faults while ensuring a replicated state machine with strong validity. We ensure that the replicated state agreed upon by the non-Byzantine participants is used when computing the next replicated state, despite recurring transient faults that corrupt the replicas maintained by a bounded number of them. 

Several prior works address the combination of Byzantine fault tolerance and self-stabilization in a replicated state machine protocol. 
%To mention a few, Binun \etal \cite{DBLP:conf/sss/BinunCDKLPYY16} presented the first design and implementation of a self-stabilizing Byzantine replicated state machine,  %\cite{DBLP:conf/podc/BonomiDPR15}, \cite{DBLP:conf/cscml/DolevGMS18}
 Dolev \emph{et al.}~\cite{dolev2018self} propose a self-stabilizing Byzantine-tolerant replicated state machine architecture, relying on failure detectors to manage faulty behavior. While their design also targets recovery and Byzantine resilience, it assumes transient faults are rare and isolated. In contrast, our solution explicitly tolerates recurring transient faults that may continuously affect system execution.
Georgiou \ et al.~\cite {DBLP:conf/netys/GeorgiouMRS21} propose a loosely self-stabilizing binary Byzantine consensus protocol for asynchronous message-passing systems without relying on digital signatures. Their work focuses on binary input domains, whereas our approach targets multivalued consensus with stronger interval validity guarantees under recurring transient faults.
Browenstein \etal \cite{DBLP:journals/tdsc/BrownsteinDK22} presented a self-stabilizing Byzantine replicated state machine that preserves the global state privacy, based on secure multi-party computation. 

In the scope of Blockchain, Dolev \ et al.~\cite {DBLP:journals/iandc/DolevL22} present techniques for using wallet information to reconstruct corrupted (or even erased) Blockchain records.
Georgiou \etal~\cite{DBLP:conf/sss/GeorgiouRS23} present a self-stabilizing, Byzantine-tolerant framework for recycling single-shot consensus objects in long-lived computations, allowing the reuse of object identifiers and other finite resources. 
Recently, Duvignau \etal~\cite{duvignau2023self,DBLP:journals/tcs/DuvignauRS25} present self-stabilizing versions of reliable broadcast and multivalued consensus protocols designed for Byzantine environments. Like our work, their approach addresses fault tolerance from arbitrary initial states and supports recovery in the presence of transient faults. However, it does not account for the recurring nature of such faults, which is a central focus of our system settings.

While previous works focus on resource reuse and managing global variables under adversarial conditions, our work targets agreement and state machine replication with recurring transient fault containment and strong validity on the replicated state guarantees.

\noindent
{\bf Paper roadmap.} System settings appear in the next section, Section \ref{sec:sys}, the requirements of the Byzantine agreement and the replicated state machine appear in Section \ref{sec:prob}. Madian-based multi-valued Byzantine agreement is presented in Section \ref{sec:med}, then our replicated state machine appears in Section \ref{sec:btfr}. Lastly, conclusions appear in Section \ref{sec:conc}. Some details appear in the Appendix, and some are omitted from this extended abstract.

\section{System Settings}
\label{sec:sys}

The Byzantine Agreement problem was first defined in \cite{DBLP:journals/jacm/PeaseSL80}. The agreement task involves a set of processes that attempt to agree on a value by executing a common protocol during which they communicate. In Byzantine agreement, some processes act maliciously, sending the worst possible sequence of (inconsistent) messages. The Byzantine participants can be two-faced, sending or omitting conflicting messages to different participants or fainting crashes.

We assume a fully synchronous system model, dividing time into discrete, globally synchronized rounds.
A global pulse, implemented via an external timing mechanism, defines the start of each round.
All non-faulty processes send their messages simultaneously at the beginning of the round, and all messages are delivered before the next round begins.

The \emph{processes} $p_1, \dots ,p_n$ are a set of processes that execute the same algorithm.
Each process can send messages to any other process. 
In each round, a process may send different messages to different recipients. 
Processes proceed according to global rounds. 
The communication graph is a clique: each process has a unique identifier and can communicate directly with all other processes. 
The sender’s identity is known upon receipt, either encoded in the message or inferred from the communication channel.

The system state (also called a \emph{configuration}) at any time consists of the vector of the local states of all processes.
Due to synchronous delivery, message buffers are assumed empty at round boundaries.
We define a system execution as a sequence of global configurations resulting from atomic computation steps executed in rounds.
Each round triggered by a global pulse consists of four phases: 
(i) an \emph{input phase}, in which every process receives an external input value simultaneously; 
(ii) a \emph{send phase}, where each process broadcasts messages; 
(iii) a \emph{receive phase}, where all messages are delivered; and 
(iv) a \emph{computation phase}, where each process updates its local state based on received messages and input.

\begin{figure}[t!]
%\begin{wrapfigure}{r}{6.15cm}
%\vspace*{-0.5cm}
\centering
\begin{tikzpicture}[font=\small, every node/.style={align=center}, layer/.style={draw, minimum width=6cm, minimum height=1.2cm, fill=gray!10}, arrow/.style={->, thick}]

% Layer positions (bottom to top)
\node[layer] (layer1) at (0,0) {Synchronous\\Message-Passing System};
\node[layer, above=0.8cm of layer1] (layer2) {Multi-valued Byzantine Agreement};
\node[layer, above=0.8cm of layer2] (layer3) {State Machine Replication};

% Arrows between layer1 and layer2
\draw[arrow] ([xshift=-2cm]layer2.south) -- node[left] {send} ([xshift=-2cm]layer1.north);
\draw[arrow] ([xshift=2cm]layer1.north) -- node[right] {receive} ([xshift=2cm]layer2.south);

% Arrows between layer2 and layer3
\draw[arrow] ([xshift=-2cm]layer3.south) -- node[left] {input} ([xshift=-2cm]layer2.north);
\draw[arrow] ([xshift=2cm]layer2.north) -- node[right] {output} ([xshift=2cm]layer3.south);

\end{tikzpicture}
\caption{\label{fig:arch}Layered architecture showing interactions between the system components.}
%\end{wrapfigure}
\end{figure}
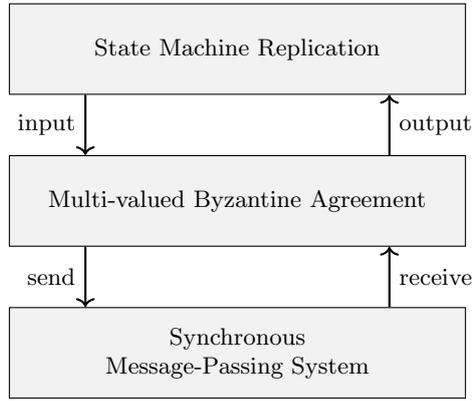

We consider \emph{Byzantine faults}~\cite{DBLP:journals/jacm/PeaseSL80}, where faulty processes may arbitrarily deviate from the algorithm and act maliciously.
To enable system recovery and continuous functionality of the system~\cite{DBLP:journals/jacm/PeaseSL80,DBLP:conf/stoc/Ben-OrGW88}, the number of Byzantine processes is bounded throughout execution by $t < \lfloor n/3 \rfloor$.

We also account for \emph{transient faults}, as introduced in Dijkstra's seminal work on self-stabilization~\cite{DBLP:journals/cacm/Dijkstra74}.
These faults occur only before the start of the execution, possibly due to temporary violations of fault thresholds, and may arbitrarily corrupt the initial system configuration.
From that point on, the fault thresholds, \eg number of Byzantine processes, are assumed to remain within tolerated bounds, and no transient fault can arbitrarily change the system state.

In contrast, we define \emph{recurring transient faults} (also called \emph{online transient faults}) as arbitrary local state corruptions that may occur throughout execution.
These affect at most one process per corruption event and are assumed to occur with bounded intensity.
Specifically, the number of such corruptions in any given round is strictly fewer than $r < \lceil n/6 \rceil-1$.

The bound on $r$ ensures that the median used in our agreement algorithms remains within the range of correct values.
Since up to $\lfloor n/3 \rfloor$ processes may be Byzantine and some correct values may be lost or unavailable, allowing more than $\lfloor n/6 \rfloor-1$ recurring corruptions could cause the median to reflect faulty or adversarial values.
This threshold ensures that a sufficient number of uncorrupted, correct inputs influence the outcome.

Due to the presence of transient faults, the system may begin in an arbitrary configuration, meaning that transient faults may corrupt each process's state. 
Our algorithms are designed to be self-stabilizing, \ie they converge to correct behavior even in the presence of a bounded (\eg one third of the processes) number of Byzantine processes together with a bounded (\eg one sixth of the processes) recurring state corruptions.

\smallskip
\noindent
\textbf{Legitimate States and Self-Stabilization.~~}
A system state is said to be \emph{legitimate} if it satisfies all requirements specified by the problem definition (\Cref{sec:prob}).
An algorithm is \emph{self-stabilizing} if, from any arbitrary state, the system is guaranteed to reach a legitimate configuration in a finite number of rounds and remains in legitimate configurations as long as the number of faults per round remains within tolerated thresholds.
Our algorithms stabilize even when starting from an arbitrarily corrupted state due to transient faults, and continue to function correctly under a bounded number of Byzantine processes and recurring state corruptions.

We define the \emph{stabilization time} as the worst-case number of synchronous rounds required for the system to reach a legitimate configuration from an arbitrary state, assuming up to $t$ Byzantine faults and up to $r$ recurring state corruptions per round.

We summarize the system description and main terms used in the sequel.
\begin{description}
    \item[Processes] \label{def:proc} $p_1, \dots ,p_n $ are a set of processes that execute the same protocol. Each process can send a message to any other process. Processes execute the protocol according to sequential global pulses. Messages are sent when a pulse occurs and are received before the next pulse appears. The communication graph is a clique. Each message is identified by the communication link through which it arrives. Therefore, the identity of the process that sent a message is known to the receiver.
    \item[Pulse number] A global pulse is associated with a global counter modulo a given integer $b$. The processes are exposed to the global pulse and the global pulse number  \footnote{The global pulse and global pulse number can themselves be an output of semi-synchronous algorithms, \eg \cite{DBLP:journals/jacm/DolevW04, DBLP:journals/jacm/LenzenR19}.}. We denote a global pulse that triggers the transition of the replicated state machine as $Pulse$, each such $Pulse$ is invoked when the global counter of pulses is zero. Then we have intermediate pulses, for which the pulse counter value is $1,2,\ldots,b-1$, denoted $IPulse$ to finalize the needed computation for the next replicated state. In the sequel, we may omit mentioning the $IPlus$ that controls the advance in the intermediate round stages whenever there is no possible confusion.
    %\SD{Elad, yours}} 
    \item[Faulty (Byzantine) Process] A faulty process is a process that deviates from the protocol. This includes processes that crash during execution or have unlimited computation power to send the most malicious sequence of messages to prevent the other processes from reaching an agreement. There are at most $f$ faulty processes where $f< \lceil n/3 \rceil - 1$.
    \item[Transient Faults] A (spontaneous) change of a state of a process to an arbitrary state in the domain of its states. There are at most $r$ transient faults where $r < \lceil n/6 \rceil -1$ in situations with no prior knowledge of the distribution of transient faults.
\end{description}

\section{Problem Statement and Solution Framework}
\label{sec:prob}

We study a class of agreement problems for systems that tolerate both recurring transient faults and Byzantine faults, provided that the number of Byzantine processes is less than $\lfloor n/3 \rfloor$ and the number of recurring transient faults is less than $\lceil n/6 \rceil - 1$.
We focus on three variants of Multi-valued Byzantine Agreement (MVBA). 
Our main contribution is a protocol that achieves \emph{interval validity} and consistency, building upon an algorithm that guarantees only \emph{weak validity}. 
For comparison, we also recall the definition of MVBA with \emph{strong validity}, and briefly mention the related concept of \emph{approximate agreement}.

\Cref{fig:arch} illustrates the layered architecture of our solution framework. At the bottom lies the synchronous message-passing system (\Cref{sec:sys}), which provides reliable, round-based communication among processes, activated by global pulses. 
On top of this network layer, the MVBA layer (Definition~\ref{def:interval-validity}, \Cref{alg:baa}) implements agreement primitives that tolerate Byzantine faults and recurring transient faults. This layer combines a multi-valued agreement structure with most-common and median-based selection mechanisms to ensure consistency and interval validity. It interacts with the communication layer via the \texttt{send} and \texttt{receive} interfaces, and with higher layers via synchronous \texttt{input} and \texttt{output} operations.
At the top of the stack is the \emph{State Machine Replication} layer (Definition~\ref{def:smr}, \Cref{alg:smr}), which relies on the agreement abstraction to execute deterministic transitions consistently across all correct replicas.

We clarify that, for the sake of a simple presentation, the studied problems are specified only by the requirements critical to understanding our contribution. For a more comprehensive treatment of consensus problems in distributed systems, we refer the reader to e.g, ~\cite{DBLP:books/sp/Raynal18, DBLP:books/mk/Herlihy2013, DBLP:books/mk/Lynch96, DBLP:books/daglib/0017536}.
It is also known that in synchronous message-passing systems, any consensus algorithm that terminates within a bounded number of rounds can be transformed into a self-stabilizing algorithm that recovers within the same bound after the occurrence of the last transient fault. See, for example, Theorem~2 of Georgiou \etal~\cite{DBLP:conf/sss/GeorgiouRS23}. 
Therefore, from this point onward, our work focuses exclusively on protocols that tolerate both recurring transient faults and Byzantine faults.

\subsection{Multi-valued Byzantine Agreement (MVBA)}

Our solution framework uses the problem specified in Definition~\ref{def:interval-validity} to construct a replicated state machine (Definition~\ref{def:smr}) that tolerates both Byzantine faults and recurring transient faults.

\begin{definition}[Reference Problem: MVBA with Strong Validity]
	\label{def:agreement}
	A set of $n$ processes executes a Multi-valued Byzantine Agreement algorithm. For all $1 \leq i \leq n$, each process $p_i$ starts with an input value $v_i \in V$ from some domain of values $V$ and returns a decision value $d_i \in V$ within a known constant number of synchronous rounds. The algorithm is correct if it satisfies the following properties:
	
	\begin{itemize}
		\item \emph{Consistency:} All non-faulty processes decide on the same value.
		\item \emph{Strong Validity:} All non-faulty processes decide on a value proposed by a non-faulty process.
	\end{itemize}
\end{definition}

A variation on Definition~\ref{def:agreement} substitutes the strong validity requirement with \emph{discrete interval validity} (Definition~\ref{def:interval-validity}), which captures the intuition that the decided value must lie within the interval of values proposed by non-faulty processes, even if a Byzantine process proposes it.

\begin{definition}[Target Problem: MVBA with Discrete Interval Validity]
	\label{def:interval-validity}
	A Multi-valued Byzantine Agreement algorithm satisfies \emph{discrete interval validity} if, in addition to satisfying consistency (Definition~\ref{def:agreement}), all non-faulty processes decide on the same value $d \in V$, and this integer value lies within the range of the inputs proposed by non-faulty processes.

	Formally, if $T = \{v_i \mid p_i \text{ is non-faulty}\}$ is the multiset of inputs from non-faulty processes, then $d \in \{\min T, \ldots, \max T\}$.
\end{definition}

An MVBA algorithm with discrete interval validity accepts values that Byzantine processes may have proposed, as long as they fall within the range of correct inputs and do not violate agreement.
The value chosen does not result from the (not necessarily integer) mean computation used to approximate the input value of the non-faulty processes, see \eg \cite{DBLP:journals/jacm/DolevLPSW86}.
Unlike weak validity (defined next), discrete interval validity does not introduce an artificial default value $\bot$.

In the sequel, we use interval validity instead of discrete interval validity whenever no confusion can arise. 

\subsection{Our Building Block: MVBA with Weak Validity}

Our solution framework for MVBA with interval validity builds upon a Byzantine Agreement algorithm that guarantees only \emph{weak validity}, such as the one proposed by Turpin and Coan~\cite{TURPIN198473}. Note that weak validity permits a decision on the default value $\bot$, which may not be acceptable for application-level correctness or progress.

\begin{definition}[Building Block: MVBA with Weak Validity]
	\label{def:weak-validity}
	A Multi-valued Byzantine Agreement algorithm satisfies \emph{weak validity} if all non-faulty processes decide on the same value (\ie satisfy consistency from Definition~\ref{def:agreement}), and the decided value is either:
	\begin{itemize}
		\item a value proposed by a non-faulty process, if all such processes proposed the same input; or
		\item a value from the decision domain $V$, which may include a default value $\bot$ otherwise.
	\end{itemize}
\end{definition}

A related but distinct problem is \emph{approximate agreement}~\cite{DBLP:journals/jacm/DolevLPSW86}, in which the decision values of non-faulty processes must be within a small $\epsilon$ range of each other and lie within the bounds of the correct inputs.

\subsection{Our Application: State Machine Replication}
\label{sec:SMRdef}

The state machine is a foundational model of computation with many variations. In general, a state machine $M$ has a set of states $q_1, q_2, \dots$ and a transition function $\delta(q, \mathit{in})$ that takes a state $q$ and an input $\mathit{in}$ and returns the resulting state. A state machine is deterministic if $\delta$ always returns the same result for the same input and state.
 
The State Machine Replication problem was introduced by Lamport in \cite{DBLP:journals/cacm/Lamport78} and later 
investigated under different settings.
%formalized in \cite{DBLP:conf/replication/SchneiderZ10}  in the context of distributed systems prone to failures, 
In \cite{DBLP:journals/csur/Schneider90}, the authors specify the state machine replication problem in distributed systems prone to failures. A distributed protocol implementing state machine replication should verify the following two properties:

\begin{description}
\item [Agreement] Every non-faulty copy of the state machine receives every command;
\item [Order] Commands are processed in the same order by every non-faulty copy of the state machine.
\end{description}
 
A \emph{replicated state machine} consists of multiple identical replicas of $M$ that maintain synchronized states and execute transitions based on agreed inputs.
The following is a formal definition of a replicated state machine based on the definitions in \cite{DBLP:conf/cscml/DolevGMS18, 10.1145/98163.98167}. 

\begin{definition}
	\label{def:smr}
	\begin{description}
		\item[State Machine:] Let $M$ be a deterministic state machine. Let $S = \{s_1, s_2, \dots\}$ be its set of states, $\Sigma = \{\mathit{in}_1, \mathit{in}_2, \dots\}$ be its input alphabet, and $\delta: S \times \Sigma \rightarrow S$ be the transition function.
		
		\item[Replicas:] The processes $p_1, \dots, p_n$ that maintain local copies of $M$ and apply transitions using $\delta$.
		
		\item[Program:] A program is a sequence of atomic steps, each consisting of a local computation followed by a communication step.
		
		\item[Configuration and Global Input:] A \emph{configuration} is defined by a function over the replica states (\eg the most common state), and the \emph{global input} is defined by aggregating local inputs (\eg the median of all inputs).
		
		\item[Execution:] A replicated execution is an alternating sequence $R = c_1, \mathit{in}_1, c_2,$ $\mathit{in}_2,$ $\dots$ such that for all $i > 0$, $\delta(c_i, \mathit{in}_i) = c_{i+1}$.
	\end{description}
\end{definition}

See \eg \cite{DBLP:journals/iandc/BazziH22} for the importance of fault-tolerant replicated state machines in the scope of Blockchain technologies.

\section{Median-based  Multi-valued Byzantine Agreement}
\label{sec:med}
We design a one-shot Multi-valued Byzantine Agreement algorithm that returns the most frequently occurring value among the agreed-upon inputs, provided that value's observed frequency surpasses the threshold required to guarantee validity. Otherwise, the algorithm returns the median value.
In cases where applications tolerate a default value $\bot$, it is always trivially possible to return that value. However, we are interested in applications where default values would not be satisfactory.
Our algorithm works by first reaching agreement on a vector of inputs using an agreement algorithm that satisfies a weaker validity condition. After this step, each non-faulty process obtains an identical vector of values. The process then removes all default entries ($\bot$), sorts the resulting list—with repetitions preserved—and selects the median. If the list contains an even number $\ell$ of non-default elements, the algorithm returns the entry at position $\ell/2$ in the sorted list. We introduce a tunable threshold parameter $\alpha$, which controls when a value is considered sufficiently frequent to be selected directly. In situations where assumptions about the distribution of values of transient faults can be made, reducing the parameter $\alpha$ allows for greater transient fault tolerance, up to a third of all processes in the best case. When no assumptions can be made about the distribution of values of transient faults, setting $\alpha$ to a sixth of all processes yields the greatest transient fault tolerance. 

In the context of replicated state machine applications, we aim to prefer the most common value in the decision vector, provided it occurs frequently enough to ensure strong validity. Specifically, if a value appears in more than one-third of the entries plus an additional margin $\alpha$, it is selected. Otherwise, the median value is used. The parameter $\alpha$ accounts for the possible influence of recurring transient faults, allowing the algorithm to distinguish between values that reflect genuine consensus and those that may result from adversarial skew.

\subsection{Algorithm description}

We now present Algorithm~\ref{alg:baa}, which implements the median-based multi-valued Byzantine agreement. The algorithm begins by allocating an array $A$ to hold values received from other processes (line 13). Each process then broadcasts its own input value to all others (lines 14-15) and records received values in $A$ (lines 16-17). Since messages include sender identities, each process can associate received values with sender indices, ensuring that $A[i]$ stores the value claimed by process $p_i$.

%Note that the Byzantine Agreement protocol of Turpin and Coan \cite{TURPIN198473} appears in the Appendix, for the convenience of the reader, copied verbatim to the first part of the Appendix, and then written using pseudocode as Algorithm \ref{alg:tce}.

Each process then invokes a weak-validity multi-valued Byzantine agreement algorithm on every entry in $A$, running these instances in parallel (lines 18-19). The algorithm $\FuncSty{WeakMVBA}$ is used here, which can be instantiated with the algorithm from Turpin and Coan~\cite{TURPIN198473}, listed in the Appendix for completeness. 
After this step, each non-faulty process holds an array $A$ in which each entry $A[i]$ reflects the agreed-upon value for process $p_i$'s input. Since the same agreement algorithm is applied to each index, and all non-faulty processes use the same inputs and follow the same logic, their resulting arrays are identical. The process then calls $\FuncSty{select\_value}(A)$ (line 20) to determine the final decision value, to be returned.

The function $\FuncSty{select\_value}$ processes the agreed-upon vector $A$ to determine the final output. It first removes all occurrences of the default value $\bot$ to produce a cleaned array $A_{\not\bot}$ (line 2), then counts the frequency of each remaining value using a dictionary $C$ (lines 4–6). The most common value $m$ is identified via an $\arg\max$ operation (line 7). If $m$ appears in at least $\lfloor k/3 \rfloor + 1 + \alpha$ entries (line 8), it is returned as the final decision. Otherwise, the array $A_{\not\bot}$ is sorted (line 10), and the median element—specifically, the $\lfloor k/2 \rfloor$-th entry—is returned (line 11). This rule ensures consistency across all non-faulty processes and protects against biased skewing from faulty or corrupted inputs.

The parameter $\alpha$ serves as a tunable resilience margin that reflects assumptions about the distribution and coordination of recurring transient faults. In adversarial scenarios, where faults may intentionally reinforce each other to skew the output, a conservatively set $\alpha = \lceil n/6 \rceil - 1$ ensures correctness despite worst-case behaviour. In more benign environments, such as when transient faults follow a uniform or independent distribution, it is often safe to choose a smaller $\alpha$, since the likelihood of multiple faults supporting the same incorrect value is significantly reduced. 
% Thus, the choice of $\alpha$ allows the algorithm to flexibly balance between robustness and decision precision, depending on the assumed fault model and application needs.
	
As noted earlier, we instantiate $\FuncSty{WeakMVBA}$ to be the algorithm of Turpin and Coan~\cite{TURPIN198473}, which returns a default value $\bot$ when agreement cannot be reached. However, the specific choice of agreement algorithm is not essential—any one-shot MVBA algorithm satisfying weak validity (\ie agreeing on a correct input when all non-faulty processes begin with the same value, and otherwise deciding on a value from the decision domain including $\bot$) is sufficient. This modularity enables reuse and flexibility in system design.
	
Our decision rule then selects the most frequently occurring agreed-upon value—excluding $\bot$—if it appears at least $\lfloor n/3 \rfloor + 1 + \alpha$ times. Otherwise, the rule returns the median value from the sorted array $A_{\not\bot}$. In case of ties (\eg even-length arrays), we break them deterministically by selecting the lower of the two middle values.
	
The returned value by the Algorithm \ref{alg:baa} is always consistent across the processes and satisfies the weak validity; proof provided next.\\
%in Section 3.2. \\

\begin{algorithm}[p]
\begin{small}
       \tcc{accepts a vector of inputs and returns one item from the vector, never returns $\bot$}
    \SetKwFunction{SelVal}{select\_value}
    \Fn{\SelVal{$A$}} {
        \KwData{$0\leq \alpha < \lceil n/6 \rceil - 1$ a parameter reflecting recurring transient fault distribution}
        \tcc{Create new array $A_{\not\bot}$ as a copy of $A$ with $\bot$ values removed}
        $A_{\not\bot} \gets \FuncSty{remove\_bottom}(A)$\;
        \tcc{After removing bottom values, we get a new array length $k \leq n$}
        $k \gets |A_{\not\bot}|$\;
        \tcc{Take $C$ to be a dictionary where keys are from $V$ and values are their frequency in $A_{\not\bot}$}
        $C \gets $ an empty dictionary\;
        \ForEach{$i \in [1 \ldots k]$} {
            $C[A_{\not\bot}[i]] \gets C[A_{\not\bot}[i]] + 1$\;
        }
        \tcc{Take $m$ to be the most common value, that is, the key with the highest value in $C$}
        $m \gets \arg\max_{v \in A_{\not\bot}}C[v]$\;
        \tcc{If the most common value appears at least $\lfloor k/3 \rfloor +1 + \alpha$ times, return it}
        \If { $C[m] \geq \lfloor k/3 \rfloor+1 + \alpha$} { 
            return $m$\;
        }
        \tcc{Else, return the median of the proposed values}
        $A^\prime \gets \FuncSty{sort}(A_{\not\bot})$\;
        return $A^\prime[\lfloor\frac{k}{2}\rfloor]$\;
    }

    \SetKwFunction{ByzAgr}{median\_byzantine\_agreement}
    \Fn{\ByzAgr{$v \in V$}}{
        \KwData{$P[1 \ldots n]$ array of addresses of all the parties}
        \tcc{Allocate an empty array to hold the initial values of all parties}
        $A[1 \ldots n] \gets \bot$\;
        \tcc{Broadcast my initial value to all parties (including self)}
        \ForEach{$i \in [1 \ldots  n]$} {
            send $v$ to $P[i]$\;
        }
        \tcc{Save the initial values received from other parties  (including self)}
        \ForEach{value $u_i$ received from party $p_i$} {
            $A[i] \gets u_i$\;
        }
        \tcc{Execute the Multi-Valued Byzantine Agreement protocol with Weak Validity in parallel to reach agreement on each value in $A$}
        \ForEach{$i \in [1 \ldots  n]$} {
            $A[i]\gets \FuncSty{WeakMVBA}(A[i])$\;
        }
        $x \gets \FuncSty{select\_value}(A)$\;
        return $x$\;
    }
    \caption{Median-based Byzantine Agreement protocol}
    \label{alg:baa}
    \end{small}
\end{algorithm}

\subsection{Correctness Proof}

We begin by proving that Algorithm \ref{alg:baa} has consistency and interval validity despite less than $\lfloor n/3 \rfloor$ Byzantine faults. This is shown in Lemmas \ref{lemma:baaconst} and \ref{lemma:baaval}. Then we further expand this result by showing that Algorithm \ref{alg:baa} remains correct despite $\lceil n/6 \rceil - 1$ transient faults in addition to $\lfloor n/3 \rfloor$ Byzantine faults, as shown in Theorem \ref{thm:baaft}. 

\begin{lemma}
\label{lemma:baaconst}
    Algorithm \ref{alg:baa} satisfies the consistency property of Definition \ref{def:agreement} despite less than $\lfloor n/3 \rfloor$ Byzantine faults.   
\end{lemma}

\begin{proof} 
Let $K$ be the set of $\lfloor\frac{2n}{3}\rfloor+1$ indexes such that $p_j$ is non-faulty for all $j\in K$ and let $A_i$ denote the local array $A$ of process $p_i$.
When $p_i$ reaches line 18, it holds that $A_i[j]=v_j$ for all $j\in K$.\\
Therefore, for any $i\in K$, due to the weak validity of $\FuncSty{WeakMVBA}$, it holds that in line 19, $\FuncSty{WeakMVBA}(A_i[j])$ returns $v_j$ since all correct participants executed the protocol with the same input value, thus no action taken by the faulty processes could change the result due to the $f$-fault tolerance of the protocol $\FuncSty{WeakMVBA}$. 

Moreover, based on the consistency of $\FuncSty{WeakMVBA}$, when all non-faulty processes reach line 20, we have that $A_j = A_i$ for all $j,i\in K$, that is, $A$ is consistent, for entries $\ell \not\in K$ too.
 
Every key inserted into the dictionary $C$ is a value from $A$ (line 6). This means that for any $m$ we choose (line 7), there exists an $i$ such that $A[i]=m$ as $m$ is a value of $C$. Therefore, since $A$ is identical in all correct processes, returning $m$ (line 9) returns the same output to all correct processes.  
\end{proof}

\begin{lemma} \label{lemma:baaval}
    Algorithm \ref{alg:baa} satisfies the interval validity property (Definition \ref{def:interval-validity}) despite less than $\lfloor n/3 \rfloor$ Byzantine faults.
\end{lemma}

\begin{proof}
We prove interval validity by showing that any value returned by the algorithm is within the interval of values of non-faulty processes. We do so by looking at the possible cases the Byzantine processes can create and showing that none violate interval validity. Since any value is either returned in line 9 or 11, we address each case separately.
Recall the set $K$ from the proof of Lemma \ref{lemma:baaconst}.
Let $T=\{v_j | j\in K\}$ be a set of values of non-faulty processes and let $I=\{x | \min T \leq x \leq \max T \}$ be the interval of the values of non-faulty processes.
By Lemma \ref{lemma:baaconst}, when all processes reach line 9, $A$ is consistent across the non-faulty processes.
Consider the possible cases for the distribution of values in $A$:
\begin{enumerate}
    \item All of the faulty processes sent values inside of $I$
    \item All of the faulty processes sent values outside of $I$
    \item Some of the faulty processes sent values outside of $I$
\end{enumerate}

\begin{itemize}
    \item In case 1, all values in $A$ are within $I$; thus, any decision value satisfies the interval validity property.
    
    \item {\it Notice that we return values only on lines 9 and 11 (which get returned again in line 20).}
    
    \item In cases 2 \& 3, 
            consider situations where $m$ is returned (line 9).
                    This only happens if at least $\lfloor\frac{n}{3}\rfloor+1$ processes send value $m$ due to the condition on line 8.
                    There are at most $f < \lfloor\frac{n}{3}\rfloor+1$ faulty processes.
                    Therefore they can only set  $\lfloor\frac{n}{3}\rfloor$ entries of $A$.
                    Thus, the faulty processes cannot force $m$ to be outside of $I$, therefore $m \in I$.
            \item {\it Notice that in line 10, $A^\prime$ is created by sorting $A_{\not\bot}$, which in turn is a copy of $A$.}
            \item In cases 2 \& 3 again, consider now situations where $A^\prime[\lfloor\frac{n}{2}\rfloor]$ (the median of $A$) is returned (line 11).
                    At most, $\lfloor n/3 \rfloor$ entries of $A$ have a value outside of $I$.
                    Therefore, there are at least $\lfloor\frac{2n}{3} \rfloor+1$ entries of $A$ with a value inside of $I$.
                    The same holds for $A^\prime$, as a copy of the values of $A$.
                    This implies that values of $A$ that are outside of $I$ are either in the top third or the bottom third, but never in the middle of $A^\prime$.
                    Therefore, the median of $A$ must be inside of $I$, that is, $A^\prime[\lfloor\frac{n}{2}\rfloor]\in I$.
\end{itemize}
\end{proof}
 
\begin{theorem}
    Algorithm \ref{alg:baa} implements the Multi-valued Byzantine Agreement specification (Definition \ref{def:interval-validity}) with interval validity in systems with less than $\lfloor n/3 \rfloor$ Byzantine processes. 
\end{theorem}

\begin{proof}
Follows from Lemmas \ref{lemma:baaconst} and \ref{lemma:baaval}.
\end{proof}

\begin{lemma}
    \label{lemma:cbat}
    Any Byzantine agreement protocol that satisfies weak validity remains consistent despite transient faults. 
\end{lemma}

\begin{proof}
    Unlike a Byzantine process, a process that experiences a transient fault broadcasts the same value to all the other processes. It can differ from the value the process computed in the previous pulse, or even be the default value $\bot$.   Therefore, the value's impact on the decision computation is consistent across the non-faulty processes. 
\end{proof}

%\section{Replicated State Machine}
\begin{theorem}
    \label{thm:baaft}
    Algorithm \ref{alg:baa} implements the Multi-valued Byzantine Agreement specification (Definition \ref{def:interval-validity}) with interval validity in systems with less than $\lfloor n/3 \rfloor$ Byzantine processes and less than $\alpha < \lceil n/6 \rceil - 1$ malicious transient faults.
\end{theorem}

\begin{proof}
    To prove consistency and interval validity still hold despite additional malicious transient faults, we consider the case in which a system of $n$ processes executing Algorithm \ref{alg:baa} experiences $\alpha$ malicious transient faults in addition to having at most $\lfloor n/3 \rfloor$ Byzantine processes. First, we show that consistency holds, then we show validity holds up to the set parameter $\alpha$ of malicious transient faults. To do so, we reexamine the cases from the proof of Lemma \ref{lemma:baaval}, this time considering malicious transient faults.

    Consistency follows from Lemma \ref{lemma:baaconst} and Lemma \ref{lemma:cbat}.
    By Lemma \ref{lemma:cbat}, every entry $A[i]$ where $i$ is the index of a process with a transient fault is consistent across the non-faulty processes.
    Consider again the three possible cases for the values of $A$ from the proof of Lemma \ref{lemma:baaval}, this time, accounting for the additional transient faults.

    \begin{itemize}
        \item In case 1, as in Lemma \ref{lemma:baaval}, interval validity property is trivially satisfied
        \item In case 2 \& 3, 
            consider situations where $m$ is returned (line 9). 
                Let $x\not\in I$ be the most common value proposed by a process with a transient fault. 
                Suppose then that $x$ is the most common value in $A$, this implies that after line 7, $m=x$. 
                Byzantine processes do not control transient processes.
                Therefore, in order to maximize the values in $A$ to be outside of $I$, all Byzatine processes must propose the value $x$.
                Thus, the frequency of $x$ in $A$ is at most $\alpha + \lfloor n/3 \rfloor$, due to the assumption of at most $\alpha$ transient faults.
                But due to the condition on line 8, $m$ will only be returned if the frequency of $m$ is greater than $\lfloor n/3 \rfloor + 1 + \alpha$.
                Therefore, faulty processes (Byzantine + transient) will always fall short of the threshold.
                Therefore, if $m$ is returned, then $m\in I$.
            
            \item In case 2 \& 3 again, consider now situations where $A^\prime[\lfloor\frac{n}{2}\rfloor]$ is returned (line 11).
                There are at most $\lfloor n/2 \rfloor - 1$ values outside of $I$ in $A$ since $\alpha < \lceil n/6 \rceil - 1$.
                Therefore there are at least $\lfloor\frac{n}{2} \rfloor + 1$ values inside of $I$ in $A$.
                {\it Notice that in line 10, $A^\prime$ is created by sorting $A_{\not\bot}$, which in turn is a copy of $A$.}
                The same holds for $A^\prime$.
                The values outside of $I$ are placed in either the bottom half or top half of $A^\prime$, but in either case, there are at most $\lfloor n/2 \rfloor - 1$ of them
                Thus $A^\prime[\lfloor\frac{n}{2}\rfloor]\in I$.
            
    \end{itemize}
\end{proof}

\begin{corollary}
    If the values yielded by recurrent transient faults are maliciously chosen, the bound on the tolerated recurrent transient faults is $\alpha=\lceil n/6 \rceil - 1$.  If the values of the recurrent transient faults are (\eg uniformly) distributed, then $\alpha$ reflects the expected number (with safety factor) of recurrent transient faults that may happen to choose a malicious value. 
\end{corollary}

We assume adversarial behavior if there is no prior knowledge regarding the distribution of transient faults. Therefore, we can tolerate at most $\lceil n/6 \rceil -1$ transient faults since the non-faulty processes must have a majority to guarantee interval validity. However, transient faults may be regarded as non-malicious; viewed as a change of a value to hold a value chosen according to a uniform distribution, their number only slightly contributes to the value chosen by the Byzantine processes.

\section{Byzantine and Transient Faults Resilient State Machine Replication}
\label{sec:btfr}

The following algorithm solves the state machine replication task in a system of \( n \) processes, where up to \( \lfloor n/3 \rfloor \) may be Byzantine and up to \( \lceil n/6 \rceil - 1 \) may experience recurring transient faults per synchronous round. The algorithm achieves this by reaching agreement on both the current state and the next input using repeated invocations of Algorithm~\ref{alg:baa}, the multi-valued Byzantine agreement with interval validity. Execution proceeds in iterations called \( \mathsf{Pulse} \)s, each consisting of several synchronous communication rounds. As mentioned in \Cref{sec:sys}, global pulses are provided by an external timing mechanism that defines the start of each round. We clarify that within each \( \mathsf{Pulse} \), several communication rounds are executed, and the external mechanism is responsible for coordinating the algorithm’s steps. Each $Pulse$ is triggered externally, and the algorithm waits for the next $Pulse$ to trigger before executing the next iteration of the main loop.

The initialization steps preceding the \texttt{Pulse} loop in~\Cref{alg:smr} are included for completeness, but they are not essential for ensuring self-stabilization. Rather, they help define the starting state when reasoning about eventual safety, in line with the discussion in~\cite{DBLP:conf/opodis/DelaetDP09}, which emphasizes the role of initial values in comparing the guarantees of self-stabilizing versus non-self-stabilizing algorithms.

\begin{algorithm}[H]
  \SetKwFunction{ByzAgr}{state\_machine\_replication}
    \Fn{\ByzAgr{input, s}}{
        \GlobalData{$M$ a state machine with states $Q=\{q_0 \dots q_m\}$, alphabet $\Sigma$, and transition function $\delta: Q \times \Sigma \rightarrow Q $}
        \tcc{Initialization for completeness, see \cite{DBLP:conf/opodis/DelaetDP09}}
        $input\_value \gets input$\;
        $current\_state \gets s$\;
        \tcc{Main loop, each iteration is externally triggered by a $Pulse$, the execution waits until the next $Pulse$}
        \Upon ($Pulse$) {
            \tcc{Agree on next input}
            $input\_value \gets \FuncSty{median\_byzantine\_agreement}(input\_value)$\;
            \tcc{Agree on current state}
            $current\_state \gets \FuncSty{median\_byzantine\_agreement}(current\_state)$\;
            \tcc{Apply state transition from state $current\_state$ with input $input\_value$ on machine $M$}
            Apply state $\delta(current\_state,input\_value)$ to $M$\;
        }
    }
    \caption{State Machine Replication protocol}
    \label{alg:smr}
\end{algorithm}
	
%\EMS{The use of \texttt{Upon Pulse} should align with the definition of \texttt{Pulse} given in Section~\ref{sec:sys}. Specifically, it should be clearly explained to the reader that this construct is triggered externally and encapsulates multiple synchronous rounds. A comment in the code of the algorithm can further help the text that appears in the start of Section 5.}
%\AH{Added explanation about it in the text above the algorithm, expaned on it in a comment in the pseudo code}

The following lemma establishes the correctness of~\Cref{alg:smr}. For simplicity, we assume that recurring transient faults may only occur at the beginning of each \texttt{Pulse}. That is, when a new instance of the Byzantine agreement protocol is invoked, each non-Byzantine process broadcasts either a correct value from its local state or a corrupted value due to a transient fault. Any other types of state deviation or faulty behavior are conservatively attributed to Byzantine processes.

\begin{lemma}
    \label{lemma:smrconst}
    Consider the execution of Algorithm \ref{alg:smr} between two successive pulses, $Pulse_i$ and $Pulse_{i+1}$, in the presence of less than $\lfloor n/3 \rfloor$ Byzantine processes and less than $\lceil n/6 \rceil - 1$ processes with corrupted state (due to recurrent transient fault), then all non-Byzantine processes reach the same state just before  $Pulse_{i+1}$. 
\end{lemma}

\begin{proof}
    Each participant $p_i$ starts the execution with some input $input_i$ and state $s_i$, which is its replica of the replicated state machine. In line 5, the participants execute Algorithm \ref{alg:baa} with $input_i$ as input; thus, by Theorem \ref{thm:baaft}, every non-faulty process has the same value in $input\_value$ after line 5. In line 6, every non-faulty process proposes a state; thus, by Theorem \ref{thm:baaft}, after line 6, every non-faulty participant has the same value in $current\_state$. 
    Due to the determinism of $\delta$, since both $input\_value$ and $current\_state$ are consistent across the non-faulty parties, $\delta(current\_state,input\_value)$ is the same for all non-faulty parties.
\end{proof}

\begin{lemma}
    Just before every pulse, $Pulse_i$, $i>1$, every non-faulty process executing Algorithm \ref{alg:smr} is in a state that is derived from a state of a non-faulty process just before $Pulse_{i-1}$ and the commonly decided input $input_{i-1}$ despite less than $\lfloor n/3 \rfloor$ Byzantine faults and less than $\lceil n/6 \rceil - 1$ transient faults.
\end{lemma}

\begin{proof}
    By Lemma \ref{lemma:smrconst}, after execution of $Pulse_i$, every non-faulty process is in the same state, denoted as $s$.
    Due to recurring transient faults, some of the non-faulty processes may change to a state different from $s$. 
    Let $r$ be the number of non-faulty processes that experience a transient fault after the execution of $Pulse_i$.
    In line 6, all non-faulty processes propose state $s$ while the $r$ transient faulty and Byzantine processes propose some other state.
    There are at most $\lceil n/6 \rceil - 1$, therefore $r < \lceil n/6 \rceil-1$.
    At least $\lfloor n/2 \rfloor + 1$ non-faulty processes, thus at least $\lfloor n/2 \rfloor + 1$ processes propose the value $s$ for the agreement step in line 6.
    Therefore, in line 7 of Algorithm \ref{alg:baa}, $m$ is chosen to be $s$ in all non-faulty processes.
    Due to the majority condition in line 8, $m$ is returned in line 9 (both lines of Algorithm \ref{alg:baa}).
    Thus, after line 6, it holds that $current\_state=s$ for every non-faulty process.  
    Theorem \ref{thm:baaft} shows that $input\_value$ is consistent across the non-faulty process.
    Therefore, every non-faulty process will apply the state $\delta(s,input\_value)$ at the end of the execution of $Pulse_{i+1}$ (line 7).
\end{proof}

We conclude the series of proofs with the following theorem that declares a constant stabilization time, too. 

\begin{theorem}
In the presence of $\lceil n/3 \rceil-1$ Byzantine participants and $\lceil n/6 \rceil -1$ recurrent transient faults, Algorithm \ref{alg:smr} satisfies the agreement and order requirements, satisfying strong validity on the state and interval validity on the inputs, from the second replicated state and onwards. 
\end{theorem}

\noindent
{\bf Benign transient faults.}
We now discuss benign, recurring, transient faults, where the outcome of the recurring transient fault is a value chosen randomly and uniformly across all possible values. The discussion regards the Byzantine agreement on the state, as the Byzantine agreement on the inputs may not be subject to recurring transient faults. Recall that the Byzantine agreement on the inputs yields interval validity using only the median criteria. 

Let $x$ be the total number of recurring transient faults in a pulse $Pulse$.
Let $y$ be the most popular value among the $z$ possible values that can be obtained due to recurrent transient faults. Let $w$ be the number of processes with value $y$. Since the recurrent transient faults are randomly chosen, $w$ is a function of $x$ and $z$. Since the $\lceil n/3 \rceil-1$ Byzantine participants can benefit by joining a value $y$ to compete with the $\lfloor 2n/3 \rfloor+1-x$ non-Byzantine that do not experience recurrent transient faults in $Pulse$, the strong validity holds when $\lceil n/3 \rceil-1+w < \lfloor n/3 \rfloor +1+ \alpha < \lfloor 2n/3 \rfloor+1-x$.

In other words, in a given $Pulse$, the $\lceil n/3 \rceil -1$  Byzantine processes may choose to join the most popular value resulting from the randomly and uniformly chosen values of recurrent transient faults. 
When the number of recurrent transient faults is slightly less than $\lceil n/3 \rceil$, (leading to a total of almost two-thirds of the processes being faulty) and the possible number of possible values of the recurrent transient faults outcomes, $z$, is large (say, exponentially larger than the number of processes) then the value chosen by the non-faulty processes has a high probability of staying the most popular value. Schemes mentioned in the related work cannot support strong validity under such severe settings. In particular, any scheme based on the mean may decide on a value (at least slightly) different from the noncorrupted state of a non-Byzantine participant.
%Further details on the case of $\alpha$ smaller than one-sixth of the processes are omitted from this extended abstract version.

Thus, we obtain the following theorem:

\begin{theorem}
When the recurring transient faults are uniformly distributed, there are $\alpha$ values that preserve strong validity on the state with high probability even when the number of recurring transient faults exceeds $\lceil n/6 \rceil$.
\end{theorem}

\section{Conclusions and Discussions}
\label{sec:conc}
We investigated self-stabilizing Byzantine-tolerant algorithms that cope with arbitrary initial configurations (due to the occurrence of any number of transient faults) and Byzantine behaviours, converging to a desired behaviour despite Byzantine and recurrent transient faults. A replicated state machine that withstands such a combination of faults and maintains a strong validity property for the state component under such severe conditions is presented. 
Interestingly, in order to verify the strong properties of self-stabilizing Byzantine-tolerant state machine replication, our implementation is constructed on top of a self-stabilizing Byzantine-tolerant agreement protocol that only needs to guarantee discrete interval validity (a weak form of validity).
Our replicated state machine transition benefits from strong validity for the state component and nondefault interval validity for the outside inputs. 
Still, when the inputs of many non-Byzantine participants are identical, a reasonable assumption in many scenarios, strong validity is also guaranteed for the outside inputs.

Practical and timely applications include Blockchain Oracle‘s services that report stock prices or currency exchange rates. Current practice is ad hoc, mostly using a combination of meetings and voting. Our infrastructure enables new such services with provable systematic guarantees.  

 Further note that the assumption of a global clock pulse may be relaxed using versions of self-stabilizing Byzantine clock synchronization algorithms, \eg \cite{DBLP:journals/jacm/DolevLPSW86, DBLP:journals/jacm/LenzenR19} that withstand recurring transient faults. As the protocol we employ \cite{TURPIN198473} uses two rounds to complete the consensus, we may start the activity of \cite{TURPIN198473} in every odd clock value. For other synchronous consensus, a similar modulo approach can be used. 

\section*{Acknowledgments}
This work was supported by the Google Research Grant, the Rita Altura Trust Chair in Computer Science, the BGU Data Center, the Frankel Center for Computer Science, and the Israeli Science Foundation (Grant No. 465/22).
This work was also supported by Vinnova (Strategic Vehicle Research and Innovation, FFI) under project 2024‑03687, ‘MAGIC: Meeting Automotive Liability Challenges through Forensics Soundness’.

%\bibliographystyle{plain}
%\bibliography{ref}

\newpage

\appendix

\section{Appendix}
\subsection{Byzantine Agreement protocol of Turpin and Coan \cite{TURPIN198473}}

Below is the exact verbatim definition of the protocol in \cite{TURPIN198473}.

\begin{quote}
    Each process sends its initial value to every other process in the first round. A process is said to be perplexed if, in the first round, it receives at least as many $(P-T)/2$ initial values different from its own. Processes that are not perplexed are said to be content. Each perplexed process sends a message to every other process in the second round. The semantics of this message are just "I am perplexed".\\

    Each process maintains three local variables: two arrays indexed by the process number and a boolean. These variables are assigned values during the first two rounds. For process $j$, and $i\neq j$, these variables are defined as follows: \\

    \begin{itemize}
        \item[$v(j)$] The processe's initial value.
        \item[$v(i)$] The initial value received from process $i$
        \item[$p(j)$] A boolean that is set true if and only if process $j$ is perplexed, that is, $v(j) \neq v(i)$ for at least as many as $(P-T)/2$ distinct values of $i$
        \item[$p(i)$] A boolean that is set true if and only if process $i$ sent a message claiming it is perplexed
        \item[alert] A boolean that is set true if and only if at least as many as $P-2T$ elements of $p$ are true 
    \end{itemize}

    The binary computation is used to reach an agreement on alert. If the binary computation agrees alert = true, there are correct processes with different initial values from $V$. In this case, all correct processes use a predefined default value from $V$ as the result of the extended computation. If agreement is alert = false, then all correct content processes have the same initial value from $V$. This value is the result of the extended computation. Perplexed processes deduce this result by using the initial value that is common to a majority of the content processes. Each perplexed process tabulates as votes the values $v(j)$ for which $p(j)$ is false. The majority vote is for the value favored by the correct content processes.
\end{quote}

Next, we present a pseudocode format for the algorithm of \cite{TURPIN198473}.
    
\begin{algorithm}[H]
     \SetKwFunction{TurpinCoan}{turpin\_coan\_ext}
     \Fn{\TurpinCoan{j, $v \in V$}}{
         \KwData{$P[1 \ldots n]$ array of addresses of all the parties}
         \tcc{Allocate an empty array to hold the initial values of all parties}
         $A[1 \ldots n] \gets \bot$\;
         % \tcc{Allocate an array to hold the perplexity flag of each party, the meaning of perplexity is detailed in \cite{TURPIN198473}}
         % $B[1 \ldots n] \gets False$\;
         \tcc{Counter to check perplexity}
         $c \gets 0$\;
         \tcc{Counter to check alertness}
         $d \gets 0$\;
         % \tcc{Alert flag}
         % alert $\gets$ False\;
         \tcc{Broadcast my initial value to all parties}
         \ForEach{$i \in [1 \ldots  n]$} {
             send $v$ to $P[i]$\;
         }
         \tcc{Save the initial values received from other parties}
         \ForEach{value $u_i$ received from party $p_i$} {
             $A[i] \gets u_i$\;
         }
         \tcc{Check if current party is perplexed}
         \ForEach{$i \in [1 \ldots  n]$} {
             \If{$A[i] \neq v \wedge i\neq j$} {
                 $c \gets c + 1$\;
                 \tcc{If more than $\frac{n-f}{2}$ sent values are different from the current party, then the current party is perplexed}
                 \If{$c \geq \frac{n-f}{2} $} {
                     \tcc{Broadcast the fact that the current party is perplexed to all other parties}
                     \ForEach{$i \in [1 \ldots  n]$} {
                         send message "$p_j$ is perplexed" to $P[i]$\;
                     }
                     \tcc{Stop checking for perplexity}
                     break\;
                 }
             }
         }
         \tcc{Save the perplexity flags received from other parties}
         \ForEach{message "$p_i$ is perplexed" received from party $p_i$} {
             % $B[i] \gets True$\;
             $d \gets d + 1$\;
             \tcc{If more than $n-2t$ parties are perplexed, then turn on the alert flag\cite{TURPIN198473}, that is, return $\bot$}
             \If{$d > n - 2f$} {
                 % alert $\gets$ True\;
                 return $\bot$\;
             }
         }

        alert $\gets \FuncSty{binary\_agreement}(alert)$\;
         
         \tcc{If the system is not alert, return the majority value from the received values}
         return the most common value in $A$\;
     }
     \caption{Byzantine Agreement protocol of Turpin and Coan \cite{TURPIN198473}}
     \label{alg:tce}
 \end{algorithm}

\end{document}